\documentclass[onefignum,onetabnum]{siamonline220329}

\usepackage{braket,amsfonts}

\usepackage{enumitem} 
\usepackage{array}

\usepackage{pgfplots}
\pgfplotsset{compat=1.17} 

\newsiamthm{claim}{Claim}
\newsiamremark{remark}{Remark}
\newsiamremark{hypothesis}{Hypothesis}
\crefname{hypothesis}{Hypothesis}{Hypotheses}

\newtheorem{prop}{Proposition}[section]

\newtheorem{asmp}{Assumption}

\usepackage{graphicx,epstopdf}

\Crefname{ALC@unique}{Line}{Lines}

\usepackage{amsopn}

\usepackage{xspace}
\usepackage{bold-extra}
\usepackage[most]{tcolorbox}

\colorlet{texcscolor}{blue!50!black}
\colorlet{texemcolor}{red!70!black}
\colorlet{texpreamble}{red!70!black}
\colorlet{codebackground}{black!25!white!25}

\lstdefinestyle{siamlatex}{%
  style=tcblatex,
  texcsstyle=*\color{texcscolor},
  texcsstyle=[2]\color{texemcolor},
  keywordstyle=[2]\color{texemcolor},
  moretexcs={cref,Cref,maketitle,mathcal,text,headers,email,url},
}

\tcbset{%
  colframe=black!75!white!75,
  coltitle=white,
  colback=codebackground, % bottom/left side
  colbacklower=white, % top/right side
  fonttitle=\bfseries,
  arc=0pt,outer arc=0pt,
  top=1pt,bottom=1pt,left=1mm,right=1mm,middle=1mm,boxsep=1mm,
  leftrule=0.3mm,rightrule=0.3mm,toprule=0.3mm,bottomrule=0.3mm,
  listing options={style=siamlatex}
}

\newtcblisting[use counter=example]{example}[2][]{%
  title={Example~\thetcbcounter: #2},#1}

\newtcbinputlisting[use counter=example]{\examplefile}[3][]{%
  title={Example~\thetcbcounter: #2},listing file={#3},#1}

\DeclareTotalTCBox{\code}{ v O{} }
{ %fontupper=\ttfamily\color{texemcolor},
  fontupper=\ttfamily\color{black},
  nobeforeafter,
  tcbox raise base,
  colback=codebackground,colframe=white,
  top=0pt,bottom=0pt,left=0mm,right=0mm,
  leftrule=0pt,rightrule=0pt,toprule=0mm,bottomrule=0mm,
  boxsep=0.5mm,
  #2}{#1}

\patchcmd\newpage{\vfil}{}{}{}
\begin{tcbverbatimwrite}{tmp_\jobname_header.tex}

\title{Finding the nonnegative minimal solutions of Cauchy PDEs in a volatility-stabilized market \thanks{Submitted to the editors DATE.
\funding{TI's work is supported in part by National Science Foundation grants DMS-1615229 and DMS-2008427. NTY's work is supported in part by National Science Foundation grant DMS-2038118.}}}
\author{Nicole Tianjiao Yang\thanks{Department of Mathematics, Emory University, Atlanta, GA 30322, USA (E-mail: \href{mailto:yang@pstat.ucsb.edu}{tyang31@emory.edu})} \and Tomoyuki Ichiba\thanks{Department of Statistics and Applied Probability, South Hall, University of California, Santa Barbara, CA 93106, USA (E-mail: \href{mailto:ichiba@pstat.ucsb.edu}{ichiba@pstat.ucsb.edu}) }}

\headers{Minimal solutions to Cauchy PDEs in a volatility-stabilized market}{Nicole Yang, Tomoyuki Ichiba}
\end{tcbverbatimwrite}
\input{tmp_\jobname_header.tex}

\begin{document}
\maketitle

%% ------------------------------------------------------------------
%% ABSTRACT
%% ------------------------------------------------------------------
\begin{tcbverbatimwrite}{tmp_\jobname_abstract.tex}

\begin{abstract}

The strong relative arbitrage problem in Stochastic Portfolio Theory seeks an investment strategy that almost surely outperforms a benchmark portfolio at the end of a given time horizon. The highest relative return in relative arbitrage opportunities is characterized by the smallest nonnegative continuous solution of a Cauchy problem for a partial differential equation (PDE) \cite{fernholz2010optimal}. However, solving this type of PDE poses analytical and numerical challenges, due to the high dimensionality and multiple solutions of the PDEs. In this paper, we develop a numerical method to solve the relative arbitrage problem and the associated PDE in a volatility-stabilized market, by formulating the capitalization as time-changed Bessel bridges. We present a practical algorithm and demonstrate numerical results for optimal arbitrage opportunities in a volatility-stabilized market.
\end{abstract}


% \begin{MSCcodes}
% 00A20, 00B10
% \end{MSCcodes}
\end{tcbverbatimwrite}
\input{tmp_\jobname_abstract.tex}
%% ------------------------------------------------------------------
%% END HEADER
%% ------------------------------------------------------------------

\section{Introduction}

\label{sec: intro}

Stochastic Portfolio Theory (SPT), introduced by Robert Fernholz \cite{fernholz2002stochastic}, provides a framework for analyzing the behavior of the portfolio and the structure of the equity markets. Unlike models in traditional portfolio theory, SPT does not rely on the existence of an equivalent (local) martingale measure, thereby allowing for the possibility of relative arbitrage opportunities. A key concept within SPT is relative arbitrage, which involves the construction of investment strategies that almost surely outperform a given benchmark portfolio over a specified time horizon, with a focus on maximizing relative returns. The optimization of relative arbitrage strategies with respect to the market portfolio is derived in \cite{fernholz2010optimal} with respect to the market portfolio. Specific market models, such as the stabilized volatility model, which demonstrate the existence of relative arbitrage opportunities, are introduced in \cite{fernholz2005relative}. Further developments in \cite{ichiba2020relative} and \cite{yang2023relative} extend the theory to account for investor interactions, where individual investors influence the dynamics of stock capitalization, and the relative arbitrage problem becomes a stochastic differential game.
% The numerical solution for these problems

The tractability of single-player or multi-player relative arbitrage models remains a challenge. 
Even in the simplified case of a single investor in a standard time-homogeneous market system, obtaining analytical solutions is notoriously difficult due to the high dimensionality introduced by the number of assets in the market. Traditional ways of solving PDEs often rely on evolution of operators along spatiotemporal grids. This poses expensive computational costs, especially for high-dimensional PDEs, where the so-called "curse of dimensionality" causes memory requirements and computational complexity to grow exponentially with the number of dimensions. Furthermore, the optimal arbitrage is modeled as the non-negative minimal solution of an associated Cauchy PDE in \cite{fernholz2010optimal}, \cite{fernholz2011optimal}, \cite{ichiba2020relative}. PDE problems with multiple solutions remain a challenge in both numerical methods and deep learning approaches \cite{di2020finding}. However, it is of great practical need to tackle these PDE problems that may have multiple solutions or to model problems as differential inequalities. There are multiple solutions in some examples of the Black-Scholes model with stochastic volatility \cite{ekstrom2010black} or pricing under financial bubbles \cite{heston2007options}. BSDEs with singular terminal conditions and their minimal nonnegative solutions are discussed in \cite{ankirchner2014bsdes}.

In this paper, we investigate the strong relative arbitrage problem in volatility-stabilized markets and derive a numerical scheme to solve its optimal arbitrage opportunities. In particular, optimal arbitrage opportunities are characterized as the minimal nonnegative continuous solution of a Cauchy PDE. We investigate the solution through a probabilistic characterization, and the market capitalization is modeled as a time-changed Bessel processes. We ensure the positivity of the capitalization processes by analyzing the logarithmic dynamics and develop an efficient interpolation scheme based on Bessel bridges to mitigate numerical instabilities. We present numerical results of the optimal arbitrage solution for a given number of stocks and demonstrate numerically the existence of nontrivial arbitrage opportunities that outperform the market portfolio.

\subsection*{Organization of the paper}
We present the numerical solution for optimal arbitrage opportunities (\cite{fernholz2010optimal}) relative to the market portfolio in the volatility-stabilized market model. We summarize the relative arbitrage and the probabilistic solution for optimal arbitrage in Section~\ref{sec:recall}. We introduce the numerical methods based on time-changed Bessel processes and Bessel bridge interpolation in Section~\ref{sec: numeric}. In addition, we provide numerical results that demonstrate the existence of relative arbitrage opportunities and approximations of the optimal arbitrage solution. Then, in Section~\ref{sec: rbsde}, we propose general grid-based methods for optimal arbitrage, considering the associated backward stochastic differential equation of the Cauchy PDE for optimal arbitrage.
 
\section{Relative arbitrage}

For a given finite time horizon $[0,T]$, we consider an admissible market model $\mathcal{M}$ consisting of a given standard Brownian motion $W(\cdot) := (W_1(\cdot), \ldots, W_n(\cdot))^{\prime}$ on the probability space $(\Omega, \mathcal{F}, \mathbb{P})$. 
Here, filtration $\mathbb{F}$ represents the ``flow of information'' in the market driven by the Brownian motion $W$, that is, 
$\mathbb{F} = \{\mathcal{F}^W(t)\}_{0 \leq t < \infty}$
and $\mathcal{F}^W(t) := \{\sigma(W(s)) ; 0 < s < t\}_{0 \leq t < \infty} $ with $\mathcal{F}^W(0) := \{\emptyset, \Omega\}$, mod $\mathbb{P}$. All local martingales and supermartingales are with respect to the filtration $\mathbb{F}$ if not written specifically.

\subsection{Market model}

We consider the market $\mathcal{M}$ contains $n$ risky assets (stocks) with capitalization $\mathcal{X}(\cdot) := (X_1(\cdot), \ldots, X_n(\cdot))^{\prime}$ driven by $n$ independent Brownian motions. Here, $^\prime$ stands for the transpose of matrices. We assume that the capitalization follows the system of stochastic differential equations below. For $t \in [0,T]$, 
\begin{equation}
\label{eq: x}
dX_i(t) = X_i(t)\big(\beta_i(\mathcal{X}(t)) dt + \sum_{k=1}^n \sigma_{ik}(\mathcal{X}(t)) dW_k(t)\big),\ \  i= 1, \ldots, n,
\end{equation}
with initial condition $X_i(0) = x_i$, $i = 1, \ldots , n$. We assume that the number of independent Brownian motions is $n$, %$\text{dim}(W(t)) = \text{dim}(\mathcal{X}(t)) = n$, 
that is, we have exactly as many randomness sources as there are stocks on the market $\mathcal{M}$. The dimension $n$ is chosen to be large enough to avoid unnecessary dependencies among the stocks we define. Here, $\beta(\cdot) = (\beta_1(\cdot), \ldots, \beta_n(\cdot))' : \mathbb{R}_+^n \times \mathbb{R}_+^n \rightarrow \mathbb{R}^n$ as the mean rates of return for $n$ stocks and $\sigma(\cdot) = (\sigma_{ik}(\cdot))_{i, k = 1, \ldots, n} : \mathbb{R}_+^n \rightarrow \text{GL}(n)$ as volatilities are assumed to be invertible, $\mathbb{F}$-progressively measurable in which $\text{GL}(n)$ is the space of $n \times n$ invertible real matrices. 
%For simplicity, denote $\mathcal{X}(t) := (X_i(t))_{i = 1, \ldots, n}$. % already defined 
To satisfy the integrability condition, we assume that for any $T > 0$,
\begin{equation}
\label{xcond}
 \sum_{i=1}^n \int_0^T \bigg (|\beta_i(\mathcal{X}(t))|+ \alpha_{ii}(\mathcal{X}(t)) \bigg)dt < \infty, 
\end{equation}
where $\alpha(\cdot) := \sigma(\cdot)\sigma'(\cdot)$, and its $i,j$ element $\alpha_{i,j}(\cdot)$ is the covariance process between the logarithms of $X_i$ and $X_{j}$ for $1\le i,j \le n$. %The market $\mathcal{M}$ is therefore a complete market. 

\begin{definition}[Investment strategy and wealth]
\label{portfoliopi}
An $\mathbb{F}$-progressively measurable process $\pi(\cdot) = (\pi_1(\cdot) , \ldots, \pi_n(\cdot) )'$ 
is called an {admissible} investment strategy if
\begin{equation}
\label{admpi}
\int_0^T (| (\pi(t))^{\prime}\beta(\mathcal{X}(t))| + (\pi(t))^{\prime}\alpha(\mathcal{X}(t))\pi(t) ) dt < \infty, \quad  \,  T\in (0,\infty),\,\, \text{a.e. } 
\end{equation}
We denote the admissible set of the investment strategy process of an investor by $\mathbb{A}$. If the admissible strategy $\pi$
takes values in the set
$$\Delta_n := \{\pi = (\pi_1,...,\pi_n)\in \mathbb{R}^n \,| \pi_1 + \ldots +\pi_n = 1\},$$
then $\pi(\cdot)$ is
a portfolio. 

An investor adopts $\pi(\cdot)\mathbb{A}$ and the wealth process $V(\cdot) = V^{\pi, v_0}(\cdot) $ of the investor follows
\begin{equation}
\label{wealth}
    \frac{dV(t)}{V(t)} = \sum_{i=1}^n \pi_i(t) \frac{dX_i(t)}{X_i(t)}, \quad  V(0) = v_{0}.
\end{equation}
\end{definition}

In particular, the {\it market portfolio} $\pi^{\mathbf{m}}$ is characterized by investing in proportion to the market weight of each stock,
\begin{equation}
\label{Mportfolio}
\pi^{\mathbf{m}}_i(t) := \frac{X_i(t)}{X(t)}, \quad i=1, \ldots, n, \quad t \ge 0,
\end{equation}
where $
X(t) = X_1(t) + \ldots + X_n(t)$, for $t \in (0,T]$, and $X(0) = x_1 + \cdots + x_n $. 
The instantaneous growth of the wealth process $V^{\mathbf{m}}(\cdot)$ amounts to the instantaneous growth of the total capitalization of the market,
\begin{equation}
\label{marketweal}
    \frac{d V^{\mathbf{m}}(t)}{V^{\mathbf{m}}(t)} = \sum_{i=1}^n \pi^{\mathbf{m}}_i(t) \cdot \frac{dX_i(t)}{X_i(t)} = \frac{d X(t)}{X(t)} , \quad t > 0, \quad V^{\mathbf{m}}(0) = v_0.
\end{equation}

\subsection{The optimization of relative arbitrage opportunities}
\label{sec:recall}
In Stochastic Portfolio Theory, the performance of a portfolio is measured with respect to a certain benchmark. 
% We recall the definition of relative arbitrage in Stochastic Portfolio Theory.
% \begin{definition}[Relative Arbitrage]
% \label{ra}
Given an investment strategy $\pi(\cdot)$ and a benchmark strategy $\rho(\cdot)$, with the same initial capital $V^{\pi}(0) = V^{\rho}(0) =1$, we shall say that $\pi(\cdot)$ represents an arbitrage opportunity relative to $\rho(\cdot)$ over the time horizon $[0,T]$, with a given $T>0$, if
$$
\mathbb{P} \big( V^{\pi}(T) \geq V^{\rho}(T) \big) = 1 \quad \text{and} \quad \mathbb{P} \big(  V^{\pi}(T) > V^{\rho}(T) \big) >0. 
$$
% \end{definition}

In this paper, we are interested in the portfolio optimization that achieves arbitrage opportunities relative to the market portfolio.
\begin{definition}[Optimal arbitrage]
\label{uTdef}
In the market system \eqref{eq: x} with given initial market capitalization $\mathcal{X}(0)$, the investor pursues the optimal arbitrage characterized by the smallest initial relative wealth
\begin{equation}
\label{u0def}
u(T) = \inf \bigg \{ w \in (0, \infty) \,  \Big \vert \, \text{ there exists } %\exists 
\pi(\cdot) \in \mathbb{A} \text{ such that } \,v = w X(0) , \, \,  V^{\pi, v}(T) \geq X(T)\bigg \} 
\end{equation}
and the relative arbitrage portfolio $\{\pi(t)\}_{0 \leq t \leq T}$ that achieves such smallest initial relative wealth $u(T)$. 
\end{definition}
More precisely, we define a sequence of subproblems $u(T-t, \cdot)$, which represents the initial optimal arbitrage quantity to start at $t \in [0,T)$ so that we match or exceed the benchmark portfolio at the terminal time $T$ of the investment horizon. For $ 0 \le t \le T$, it follows that
\begin{equation}
\label{utobj1}
\begin{aligned}
u(T-t, \mathcal{X}(t)) = \inf \bigg \{ \omega \in (0, \infty) \,  \Big \vert \, \text{ there exists } & %\exists 
\widetilde{\pi}(\cdot) \in \mathbb{A} \text{ such that } \\
&\,v = \omega \widetilde{X}(t) , \, \,  \widetilde{V}^{\widetilde{\pi}, v}(T) \geq \widetilde{X}(T) \bigg \}. 
\end{aligned}
\end{equation}
At every time $t$, the investor
optimizes $\, \widetilde{\pi}(\cdot)\, $ from $t$ to $T$, to obtain the optimal quantity as defined in \eqref{utobj1}. 

\begin{asmp}
\label{asmp1:main}
Let $X_i(t) \beta_i(t) =: b_i(\mathcal{X}(t))$, 
$X_i(t) \sigma_{ik} (t) =: s_{ik}(\mathcal{X}(t))$, $\sum_{k=1}^n s_{ik}(t)s_{jk}(t) =: a_{ij}(\mathcal{X}(t))$ for $t \ge 0 $.
% \begin{enumerate}[label=\alph*.]
% \item \label{asmp1alipfn}
% Assume that the Lipschitz continuity and linear growth condition are satisfied with Borel measurable mappings $b(\mathbf{x})$ and $s(\mathbf{x})$. 
% That is, there exists a constant $C_L, C_G \in (0, \infty)$ that is independent of $t \in [0,T]$, such that
% \begin{equation}
% \label{bslip}
%     |b(\mathbf{x}) - b(\mathbf{x}')| + |s(\mathbf{x}) - s(\mathbf{x}')| \leq C_L |\mathbf{x}-\mathbf{x}'|, \quad \|b(\mathbf{x})\| + \|s(\mathbf{x})\| \leq C_G(1+\|\mathbf{x}\|),
% \end{equation}
% for any $\left( \mathbf{x}, \mathbf{x}' \right) \in \mathbb{R}_+^n \times \mathbb{R}_+^n$.
We assume that the market price of the risk process $\theta: [0, \infty) \times \mathbb{R}^n \times \mathbb{R}^n \rightarrow \mathbb{R}^n$ exists and is square-integrable. That is, there exists an $\mathbb{F}$-progressively measurable process such that for any $t \in [0, \infty)$,
\begin{equation}
\label{assum}
\sigma(\mathcal{X}(t)) \theta(\mathcal{X}(t)) = \beta(\mathcal{X}(t)), \ \ \mathbb{P} \bigg( \int_0^T ||\theta(\mathcal{X}(t))||^2 dt < \infty, \forall T \in (0, \infty) \bigg) = 1.
\end{equation}
\end{asmp}

We define the deflator based on the market price of the risk process $L(t)$ as
\begin{equation}
    \label{ltmg}
dL(t) = - \theta(t) L(t) dW_t, \quad t \ge 0, \quad L(0) = 1. 
\end{equation}
% Equivalently,
% \[
% L(t) := \exp \Big\{- \int_0^t \theta'(s) dW(s) - \frac{1}{2} \int_0^t ||\theta(s)||^2 ds \Big\}, \quad 0 \leq t < \infty. 
% \]
The market is endowed with the existence of a local martingale $L(\cdot) $ with $\mathbb{E}[L(T)] \leq 1$ under Assumption~\ref{asmp1:main}.
It has shown in \cite{fernholz2010optimal} that with Markovian market coefficients, $u(T)$ in \eqref{u0def} can be represented as $u(T, \mathcal X(0))$ from  
\begin{equation}
\label{conduu0}
\begin{aligned}
u(T-t, \mathcal{X}(t)) = \frac{\mathbb{E}[L(T)X(T) | \mathcal{F}_t]}{L(t)X(t)}, \quad 0 \le t \le T.
\end{aligned}
\end{equation}
% where $u(T-t, \cdot)$ is the initial optimal arbitrage quantity to start at $t \in [0,T)$ such that we match or exceed the benchmark portfolio at terminal time $T$, that is,
% \begin{equation}
% \label{utobj1}
% u(T-t, \mathcal{X}(t)) = \inf \bigg \{ \omega^{\ell} \in (0, \infty) \,  \Big \vert \, \text{ there exists } %\exists 
% \widetilde{\pi}^{\ell}(\cdot) \in \mathbb{A} \text{ such that } \,v^{\ell} = \omega^{\ell} \widetilde{X}(t) , \, \,  \widetilde{V}^{{\ell}}(T) \geq e^{c_{\ell}} \cdot \widetilde{X}(T) \bigg \} 
% \end{equation}
% for $ 0 \le t \le T$. 
% At every time $t$, the investor 
% optimizes $\, \widetilde{\pi}(\cdot)\, $ from $t$ to $T$, to obtain the optimal quantity as defined in \eqref{utobj1}. 

\section{Numerical solution to optimal arbitrage}
\label{sec: numeric}

In \cite{fernholz2010optimal}, it is shown that the optimal quantity $u(\cdot)$ in \eqref{conduu0} is
%By \eqref{optau}-\eqref{optaue}, t
% \[
% u(T, \mathbf{x}) = \mathbb{E}[L(T)X(T)]/(X(0))
% \]
the minimal non-negative continuous solution $u \in C^{1,2}([0,T] \times \mathbb{R}^n)$ of the semi-linear parabolic Cauchy problem,
\begin{equation}
\label{cauchypde}
    \frac{\partial u}{\partial \tau}(\tau, \mathbf{x}) = \frac{1}{2} \sum_{i=1}^n \sum_{j=1}^n a_{ij}(\mathbf{x}) D_{ij}^2 u(\tau, \mathbf{x}) + \sum_{i=1}^n \sum_{j=1}^n \frac{a_{ij}(\mathbf{x}) D_{i} u(\tau, \mathbf{x})}{x_1 + \ldots +x_n},
\end{equation}
\begin{equation}
\label{cauchypdeu0}
u(0,\cdot) = 1.
\end{equation}
where for $\tau := T-t$, the initial capitalization is $\mathcal{X}(t) = \mathbf{x}$. In general, this PDE has no explicit solution and the grid-based method is also hard to compute in high dimensions. 
%\label{vsmsec1}

This section seeks to develop the numerical solution to \eqref{cauchypde} - \eqref{cauchypde} under the volatility-stabilized market model (\cite{fernholz2005relative}). In Section~\ref{bealgo}, we give a numerical solution for optimal arbitrage opportunities in the volatility-stabilized market by simulating stocks from the Bessel bridges. We discuss the numerical performance for different interpolation schemes to obtain the capitalization process at terminal time $T$.

 \subsection{Challenges in finite difference methods}
% More details are reviewed in Chapter~\ref{intro}. 
The handling of \eqref{cauchypde} - \eqref{cauchypdeu0} with finite-difference methods presents several challenges. Firstly, the grid-based numerical schemes are notoriously expensive in terms of its computational costs, known as ``curse of dimensionality''. Secondly, $u(\tau,\mathbf{x})$ takes values from an unbounded domain of $(\tau,\mathbf{x}) \in \mathbb{R}_+ \times \mathbb{R}_+^n$ and the solutions of \eqref{cauchypde} - \eqref{cauchypdeu0} are not unique. Some artificial boundary conditions of $u(\tau,\mathbf{x})$ need to be carefully chosen for implementation. However, it is a delicate issue to select the correct minimal nonnegative solution, especially with a constant initial condition \eqref{cauchypdeu0} and additional boundary conditions. Thirdly, the capitalization takes values in the positive cone $(0, \infty)^n$, and thus it should not explode or go to zero. Direct simulation $\mathcal{X}(t)$ by discretizing \eqref{eq: x} with, for example, the Euler scheme produces $\mathcal{X}(t)$ values that are inevitably very close to zero. This causes a numerical overflow when approximating $u(T-t, \mathcal{X}(t))$ by \eqref{conduu}. In numerical experiments, we found that the finite-difference solution explodes or goes to zero quite easily even when the number of stocks is set as small as $n = 2$. 

% One way to reduce these obstacles under grid-based schemes is to consider the probabilistic representation \eqref{conduu}, where only first order derivative are required to discretize the processes $X_i(\cdot)$, $i = 1, \ldots, n$. 

\subsection{Volatility-stabilized market model (VSM)}
\label{numvsm}
The volatility stabilized market model (VSM), introduced in \cite{fernholz2005relative}, possesses similar characteristics in real markets, such as the \textit{leverage effect}, where instantaneous rates of return and volatility have a negative correlation with the capitalization of the stock relative to the market $\{\mathbf{m}_i(t)\}_{i = 1, \ldots, n}$ in \eqref{Mportfolio}. 

The market capital in volatility-stabilized model is modeled as a system of stochastic differential equation for $\mathcal X (\cdot) = (X_1(\cdot), \ldots , X_n (\cdot))$ defined by 
\begin{equation}
\label{kvsmeq}
    dX_i(t) = \kappa X(t)dt + \sqrt{X_i(t)X(t)} dW_i(t), \quad i = 1, \ldots, n,
\end{equation}
with $X(\cdot) = X_1(\cdot) + \cdots + X_n (\cdot)$,  where $n \geq 2$, $\kappa \in [\frac{1}{2}, 1]$.
This corresponds to $a_{ij} = \sqrt{X_i(t)X(t)} \delta_{ij}$ in \eqref{cauchypde}, where $\delta_{ij} = 1$ if $i=j$, and $\delta_{ij} = 0$ otherwise. Here, $W_1(\cdot), \ldots, W_n(\cdot)$ are independent standard Brownian motions. See \cite{fernholz2010optimal} for more details on the optimal arbitrage problem in volatility-stabilized models, and its relations to Bessel and Jacobi processes are discussed in \cite{goia2009bessel}.

We first define the squared Bessel process the Bessel process below.
\begin{definition}
For every $m \geq 0$ and $x \geq 0$, the unique strong solution of the equation 
\begin{equation}
\label{eq:besselQ}
    d Q_t = m dt + 2\sqrt{|Q_s|} dW_t, \quad t \ge 0 , \quad Q_0 = x,
\end{equation}
is called the square of $m$-dimensional Bessel process started at $x$. 
%and is denoted by $BESQ^{m}(x)$. 
$W_t$ is a linear Brownian motion with quadratic variation $\langle W, W \rangle _t = t$, $t \ge 0$. %Based on process $Q_t$, 
\end{definition}
It is shown in \cite{goia2009bessel} that \eqref{eq:besselQ} has a unique strong solution for all $m \geq 0$ and $x \geq 0$. The $m$-dimensional Bessel process is the positive square root of $Q_\cdot$
\[
R_t := sgn(Q_t) \sqrt{|Q_t|}, \quad t \ge 0 , \quad  R_0 = sgn(x)\sqrt{|x|}.
\]
From It\^o's formula,
\[
dR(t) = \frac{m-1}{2R(t)} dt + dW(t) , \quad t \ge 0 .
\]

\subsection{Numerical solution of optimal arbitrage in VSM}
\label{bealgo}
In this section, we show a numerical solution to optimal arbitrage opportunities in a class of volatility-stabilized markets by simulating $\{\mathcal{X}(t)\}_{t \in [0,T]}$ through time-changed Bessel processes. 

Define a continuous, non-decreasing stochastic process
\begin{equation}
    \Lambda(t) := \int_0^t \frac{\mathcal{X}(s)}{4} \, ds, \quad t \geq 0.
\end{equation}
By It\^o's formula, the square root of the capitalization process can be written as
\[
d \sqrt{X_i(t)} =  \frac{m-1}{2\sqrt{X_i(t)}} d \Lambda(t) + d\widehat{W}_i(\Lambda(t)),
\]
where $m = 4 \kappa = 2(1+\zeta)$, $\widehat{W}_i(t) = \int_0^{\Lambda^{-1}(t)} \sqrt{\Lambda'(s)} \, dW_i(s)$ for $t, s \geq 0, \; i = 1, \ldots, n,$
and $\langle \hat{W}_i, \hat{W}_j \rangle (s) = s \delta_{ij}$.
Consider the processes $R_i(t) = \sqrt{X_i(\Lambda^{-1}(t))}$, $i=1, \ldots, n$.
Here, we can understand $\Lambda^{-1}(t)$ as a stochastic clock. The squared capitalization process under the time change through this clock gives the independent Bessel processes $R_i(\cdot)$. 
which follows
% % of order $m$, $i = 1, \ldots, n$, 
\[
dR_i(s) = \frac{m-1}{2R_i(s)} ds + d\widehat{W}_i(s) , \quad s \ge 0 .
\]

% It follows that $Q_i(\cdot) = R_i^2(\cdot)$ for $i = 1, \ldots, n$ are independent squared Bessel processes with order $m$. $Q(\cdot) = Q_1(\cdot) + \ldots + Q_n(\cdot)$ is a Bessel process with order $mn$.

Given the starting time $s \in [0,T]$, the time-changed stock capitalization processes follow squared Bessel processes in $m$ dimension, 
\begin{equation}
\label{eq:yi}
    Y_i(t_k) = R_i^2(t_k) = X_i(\Lambda^{-1}(t_k)) = X_i(\Lambda^{-1}(t_0)) + \sum_{d = 1}^m \big(\sum_{\ell = 1}^k \Delta W_i^{(d)}(t_{\ell}) \big)^2
\end{equation}
with uniform mesh $t_k := s + k \Delta t$, $k \ge 1$ for $i = 1, \ldots, n$, where $\Delta t$ is the time step, $\Delta W_i^{(d)}(t_\ell) = W_i^{(d)}(t_{\ell+1}) - W_i^{(d)}(t_\ell)$ is the increment over $[t_\ell, t_{\ell+1}]$, implemented as $\sqrt{\Delta t} Z_i^{(d)}$ and $Z_i^{(d)}$ is a standard Gaussian random variable, for $i = 1, \ldots, n$, $d = 1, \ldots, m$.
%%%%%%%%%%%%%%%%%%%%%%%%%%%%%%%%%%%%%%%%%%%%%%%%%%
% \footnote{
% {\color{red} Need to add definition of $\Delta W(t_\ell)$}
% }
%%%%%%%%%%%%%%%%%%%%%%%%%%%%%%%%%%%%%%%%%%%%%%%%%%%%
The time-changed  process of the total capitalization of the stocks is $Y(\cdot) := \sum_{i=1}^n Y_i(\cdot)$ which follows 
\[
Y(t) = X_1(\Lambda^{-1}(t)) + \cdots + X_n(\Lambda^{-1}(t)) = 4 \Lambda'(\Lambda^{-1}(t)), \quad t \geq 0.
\]
Thus, the mapping of the clock $t_k \rightarrow \theta_k$ of the time change is
\[
\theta_k := \Lambda^{-1}(t_k) = \sum_{\ell=1}^k \frac{4}{ Y(t_{\ell})} \Delta t.\] 
%hence $\theta_{k+1} = \theta_k + \frac{4}{ Y(t_{\ell})} \Delta t$, $\theta_0 = s$. 
With this, we can find the required range of uniform mesh such that $\theta_{N-1} \leq T \leq \theta_N$ for an appropriate $N$. 
We can then use the results of $Y_i(\cdot)$ and $Y(\cdot)$ to estimate \eqref{conduu} by refining the last segment $[\theta_{N-1}, \theta_{N}]$. We interpolate between $(\theta_{N-1}, X(\theta_{N-1}))$ and $(\theta_{N}, X(\theta_{N}))$. We next discuss how to choose a suitable interpolation scheme.

For a quick implementation, one can choose a linear interpolation as
\begin{equation}
\label{eq: linearinterp}
    \begin{aligned}
X_i(\theta) 
&= \frac{\theta_k - \theta}{\theta_k - \theta_{k-1}} Y_i(t_{k-1}) + \frac{\theta - \theta_{k-1}}{\theta_k - \theta_{k-1}} Y_i(t_{k})
\end{aligned}
\end{equation}
with $\theta = T$. Thus, we obtain the simulation of $X_i(T)$, and similarly, $X(T)$, from the interpolation. 
A more accurate interpolation can be obtained by simulating the Bessel bridge in the last time step between $(\theta_{N-1}, X(\theta_{N-1}))$ and $(\theta_{N}, X(\theta_{N}))$.
% In the example of \eqref{cauchypde}-\eqref{cauchypdeu0}, we model it by the Bessel bridge $X$ of dimension $2$ starting at $X_{\theta_k}$ such that it finishes at $X_{\theta_{k+1}}$ at time $T = t_{k+1}$.
We construct Bessel bridges $R_b(t)$ from Brownian bridges (\cite{oksendal2013stochastic}).
% ; see \cite{oksendal2013stochastic} and a particular treatment of the Bessel bridge in the volatility-stabilized market model in \cite{yang2021topics}.
% The $m$-dimensional Bessel process is generated from the Euclidean norm of $m$-dimensional Brownian motions. 
That is, $R_b(t) = \left( \sum_{i=1}^m (H_t^i)^2 \right)^{\frac{1}{2}}$, 
where $H_t$ is the Brownian bridge from $a \in \mathbb{R}^n$ to $b\in \mathbb{R}^n$ over $[\theta_{k-1},\theta_{k}]$. 
% We have
% \[
% dH_t = \frac{b-H_t}{T-t} dt + dB_t, \quad H_{\theta_{k-1}} = a, \, H_{\theta_{k}} = b.
% \]
Let the initial value $R_b(\theta_{k-1}) := \left(Y_i(t_{k-1})\right)^{\frac{1}{2}} > 0$. By It\^o's formula,
\begin{equation}
\label{bebridge}
\begin{aligned}
d R_b(t) 
%&= \bigg( \frac{m-1}{2R_b(t)} - \frac{R_b(t)}{T-t} + \frac{z \sum_{i=1}^m Y^{(i)}_t}{R_b(t) (T-t)}\bigg) dt + \sum_{i=1}^m \frac{Y^{(i)}_t}{R_b(t)} dW_t^{(i)}\\
&= \bigg(\frac{m-1}{2R_b(t)} - \frac{R_b(t)}{T-t} + \frac{z \sum_{i=1}^m Y^{(i)}_t}{R_b(t) (T-t)}\bigg) dt + dZ_t,
\end{aligned}
\end{equation}
where the end point is $R_b(\theta_k) = Y_i(t_{k})$. 
% For $i = 1, \ldots, n$, $\{W_t^{(i)}\}$ is a sequence of independent Brownian motions. 
$Z_t$ is a standard Brownian motion.

Now, we can solve the optimal arbitrage which is the nonnegative minimal solution of \eqref{cauchypde}-\eqref{cauchypdeu0} through \eqref{conduu} and the time-changed Bessel processes. With $\kappa = 1$, \eqref{kvsmeq} reduces to the system of the stochastic differential equations 
\begin{equation}
\label{volstabstate}
    dX_i(t) = X(t)dt + \sqrt{X_i(t)X(t)} dW_i(t), \quad t \ge 0 
\end{equation}
for $i = 1, \ldots, n$. If the market follows \eqref{volstabstate}, the resulting solution $u(\cdot)$ is
\begin{equation}
\label{conduu}
u(T-t, \mathcal{X}(t)) = \frac{\mathbb{E}[L(T)X(T) | \mathcal{F}_t]}{L(t)X(t)}
= \frac{X_1(t) \cdot \cdots  \cdot X_n(t)}{X_1(t) + \cdots + X_n(t)}\mathbb{E}\bigg[ \frac{X_1(T) + \cdots + X_n(T)}{X_1(T) \ldots X_n(T)} \ \bigg| \ \mathcal{F}_t \bigg]
\end{equation}
for $t \ge 0 $. 
% , or equivalently, by It\^o's formula,
% \[
% d(\log X_i(t)) = \frac{1}{2 \mu_i(t)} dt + \sqrt{\frac{1}{\mu_i(t)}} dW_i(t), \quad t \ge 0.
% \]
At each time $t$, $\mathcal{X}(t) = \mathbf{x}$ is given. Hence, at time step $t$, we rerun the simulation of $\{\mathcal{X}(t)\}_{t \in [0,T]}$ from the given state at $t$ to satisfy the Markov property and provides the correct conditional expectations for $u(\cdot)$ in \eqref{conduu}. We implement this by the Monte Carlo scheme with $N$ different realizations of the Brownian motion, and thus $N$ sample paths with initial point $\mathcal X(s) := \{X_1(s), \cdots, X_n(s)\}$ such that
    % To solve $u_{T - s, X_{s}}$ in general, we have
%repeat step 1-4 with the corresponding $s_j$,
\begin{equation}
\label{eq:mc-u}
    \begin{aligned}
u(T - s, \mathcal X(s)) &= \frac{X_1(s) \ldots X_n(s)}{X(s)} \frac{1}{N} \sum_{j=1}^N \frac{X^{(j)}(T)}{ X^{(j)}_1(T) \ldots X^{(j)}_n(T)},
\end{aligned}
\end{equation}
where $X^{(j)}_i(\cdot)$ is the $j-$th sample trajectory of the capitalization of stock $i$.

\subsection{Algorithm and experimental results}

We summarize the numerical algorithm in Algorithm~\ref{gmm}. The implementation of Algorithm~\ref{gmm} for numerical examples is carried out in Python.

\begin{algorithm}[h]
\caption{Solve $u$ by simulating Bessel processes in VSM}
\label{gmm}
\textbf{Input:} $n = $ the number of capitalization processes in the market, $m =$ the dimension of the time-changed Bessel processes to model $X_i(t)$ for $i = 1, \ldots, n$, $N_T = $ the number of uniform meshes on $[0,T]$, $n_p =$ sample trajectories used for Monte Carlo simulation of $u(\cdot)$.

For $s = 0$ to $N_T$:
\begin{enumerate}
 \item Initialize $k = 0$, the states $\mathcal{X}(s) := (x_1, \ldots, x_n)$, $\theta_0 = s$.
        \item While $\theta_k \leq T$:
        \begin{enumerate}
            \item Set $k \gets k+1$, $t_k := s + k \Delta t$.
            \item Generate $n_p$ samples of $m$-dimensional independent Brownian Motion $W(t_k)$.
            \item Simulate $n_p$ independent samples of $mn$-dimensional squared Bessel processes $Y(t_{k}) = \sum_{i=1}^n Y_i(t_{k}),$ where $Y_i(t_{k})$ is defined in \eqref{eq:yi} for $i = 1, \ldots, n$.
                % \[
                % Y_i(t_{k}) = x_i + \bigg(\sum_{\ell = 1}^k W(t_{\ell})\bigg)^2, \quad \text{for } i = 1, \ldots, n.
                % \]
            \item Update $\theta_{k+1} = \theta_k + \frac{4}{Y(t_{k})} \Delta t$, where $\Delta t := T / N_T$.
        \end{enumerate}
        % \item Collect $\{\theta_0, \ldots, \theta_{k+1}\}$. Simulate the squared Bessel bridge $R^b_{X(\theta_k), X(\theta_{k+1})}(T)$.
        \item Evaluate $X_{i}(T)$ using interpolation techniques between the non-uniform mesh points $(\theta_k, \theta_{k+1})$. Then, compute $u(T-s, \mathcal{X}_s)$ by the Monte Carlo estimation \eqref{eq:mc-u}.

 \end{enumerate}
\textbf{Output:} The optimal arbitrage path $u(T-t, \mathcal{X}(t))$ for $t := s \Delta t$, $s = 0, 1, \ldots, N_T$.

\end{algorithm}

We consider the market models with a given number of stocks and time horizon. Figure~\ref{fig:uxcompare} shows the evaluated quantity of $u(T-t, \mathbf{x})$ along the $\mathbf{x}$ axes and the time $t$ axis, respectively.
The implementation of $\mathcal X(T)$ is done by the method specified in Algorithm~\ref{gmm} with the linear interpolation scheme \eqref{eq: linearinterp}.

We observe that small initial value of capitalization $\mathbf{x}$ can lead to numerical instabilities, as shown in the plots (c)-(d) of Figure~\ref{fig:uxcompare}. In particular, we visualize $u(T, \mathbf x)$ with $(x_1, x_2)$ taken from a mesh $[3.5, 9]^2$, discretized by 50 cells in each direction. The subfigure (c) in Figure~\ref{fig:uxcompare} takes the rest of the initial conditions $x_i = 4$ for $i=3, \ldots, 8$, and the subfigure (d) takes the rest of the initial conditions $x_i = 14$ for $i=3, \ldots, 8$.
% With the Bessel bridge interpolation, suitable numerical solver, and small discretization steps, the numerical instability  
One possible approach to solve this is the scaling property of the squared Bessel process, see \cite{carr2004bessel}, \cite{dufresne2004bessel}.

% \textcolor{red}{$n=8$ for figure 2?}
\begin{figure*}[ht]
  \centering
  % \begin{subfigure}[t]{0.25\linewidth}
  %   \centering
  %   \includegraphics[width=\linewidth]{n-2_x0rand.png}
  %   \caption{}
  % \end{subfigure}%
  % \begin{subfigure}[t]{0.25\linewidth}
  %   \centering
  %   \includegraphics[width=\linewidth]{ut_x0ones_ci.png}
  %   \caption{}
  % \end{subfigure}%
  % \begin{subfigure}[t]{0.25\linewidth}
  %   \centering
  %   \includegraphics[width=\linewidth]{n-8-x3-4.png}
  %   \caption{}
  % \end{subfigure}%
  \begin{tabular}{cccc}
    \includegraphics[width=.22\linewidth]{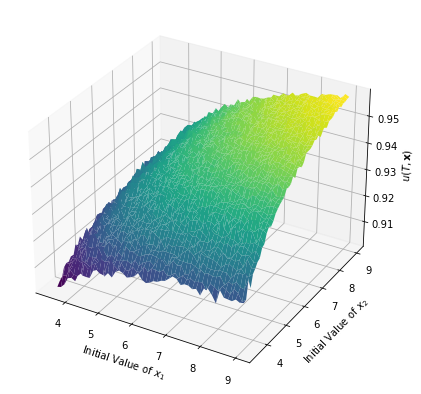}&
    \includegraphics[width=.22\linewidth]{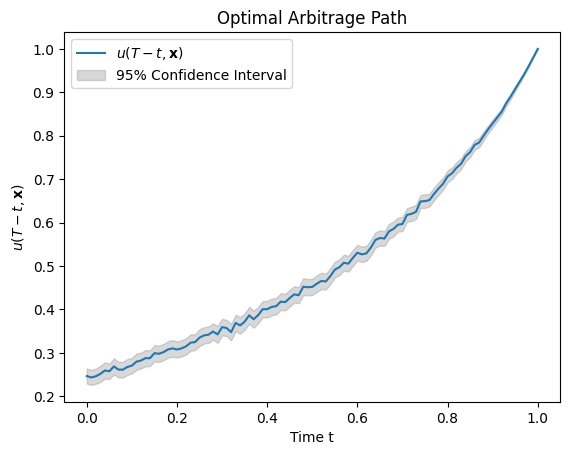}&
    \includegraphics[width=.22\linewidth]{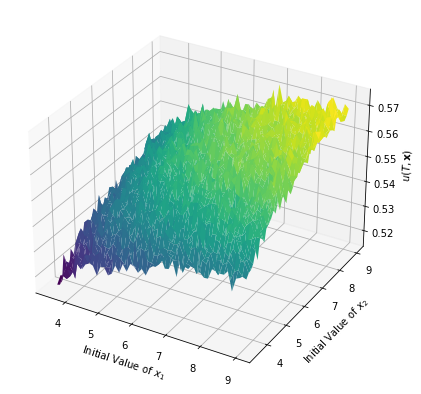}&
    \includegraphics[width=.22\linewidth]{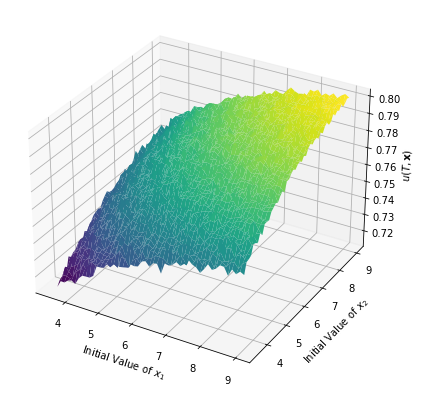}\\
    (a) & (b) & (c) &(d) \\
    \end{tabular}
  \caption{Fix $T = 1$, time step $\Delta t = 0.01$, $n_p=1000$ sample paths of Brownian motion is used to generate Bessel processes of dimension $m=4$. We solve the optimal arbitrage for the number of stocks $n=2$ in (a) - (b), and $n=8$ in (c) - (d). (a): Approximated $u(T, \mathbf{x})$ with $(x_1, x_2)$ taken from a mesh $[3.5, 9]^2$ with discretization of 50 cells in each direction. (b): Solution $u(T-t,\mathbf{x})$ for $t \in [0,T]$, with $(x_1, x_2) = (1,1)$. (c): Approximated $u(T, \mathbf x)$ with $(x_1, x_2)$ taken from a mesh $[3.5, 9]^2$ with discretization of 50 cells in each direction, $\{x_i\}_{i=3}^8$ is taken to be 4. (d): Approximated $u(T, \mathbf{x})$ with $(x_1, x_2)$ taken from a mesh $[3.5, 9]^2$ with discretization of 50 cells in each direction, $\{x_i\}_{i=3}^8$ is taken to be 14.}
  \label{fig:uxcompare}
\end{figure*}

From the interpolation by the Bessel bridge, we get the results in Figure~\ref{figuxbessel}. Higher-order numerical SDE solvers, such as the Milstein method, can be used to simulate the Bessel bridge process more accurately with numerical stability.
\begin{figure*}[ht]
  \centering
%   \begin{subfigure}[t]{0.33\linewidth}
%     \centering
%     \includegraphics[width=0.85\linewidth]{n-2_x0rand.png} 
%     \caption{}
%   \end{subfigure}%
%   \begin{subfigure}[t]{0.33\linewidth}
%     \centering
% \includegraphics[width=0.87\linewidth]{ut_bessel_ci.png}
% \caption{}
%     \end{subfigure}
%       \begin{subfigure}[t]{0.33\linewidth}
%     \centering
% \includegraphics[width=0.95\linewidth]{zeta.png}
% \caption{}
%     \end{subfigure}
\begin{tabular}{ccc}
    \includegraphics[width=0.30\linewidth]{n-2_x0rand.png}& 
\includegraphics[width=0.30\linewidth]{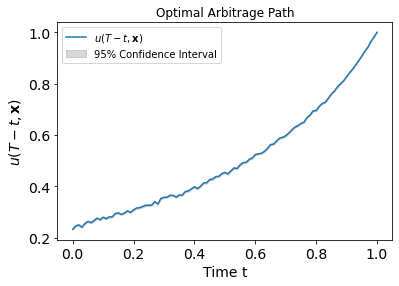}&
\includegraphics[width=0.32\linewidth]{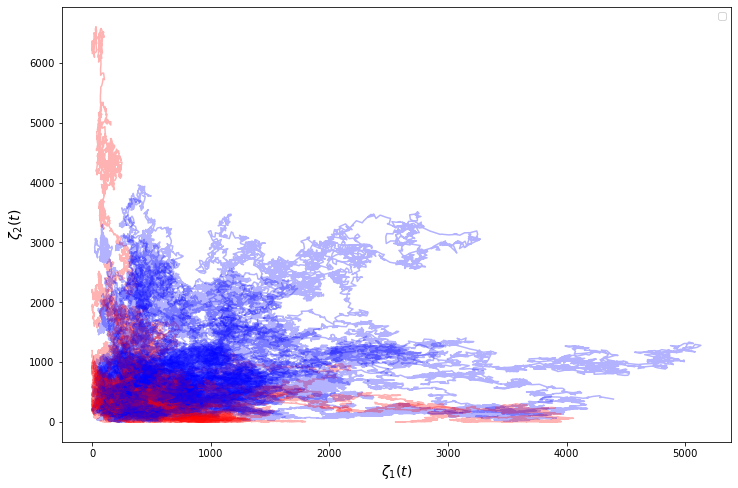}\\
(a) & (b) & (c) 
    \end{tabular}
    \caption{Consider $n=2$ stocks, time horizon $[0,T]$ with $T = 1$, and time step $\Delta t = 0.01$. We use $n_p=1000$ sample paths of Brownian motions to generate Bessel processes of dimension $m=4$. The interpolation with Bessel bridges is approximated through Euler-Maruyama of \eqref{bebridge} with time increment $0.0001$. (a): Approximated solution $u(T, \mathbf{x})$ with $(x_1, x_2)$ taken from a mesh $[3.5, 9]^2$ with spatial discretization of 50 cells in each direction. (b): Approximated solution $u(T-t,\mathbf{x})$ for $t \in [0,T]$, with $(x_1, x_2) = (1,1)$. The comparison of computing with (left) or without (right) the initial condition $u(0, \mathbf{x}) = 1$.  (c): Sample trajectories of the $n$ dimensional auxiliary process $(\zeta_1(\cdot), \ldots, \zeta_n(\cdot))$ over time horizon $[0,1]$. The sample trajectories that touch the boundary $\{\zeta_i = 0\}_{i=1}^n$ is marked red. The trajectories that stay positive is marked blue.}
    \label{figuxbessel}
\end{figure*}

To this end, we consider the existence of non-trivial relative arbitrage. In this paper, this means that the relative arbitrage should be strictly less than one. Otherwise, an investor will replicate exactly the market portfolio and match the market capitalization with $u(\cdot, \mathbf{x}) \equiv 1$.
It is shown in \cite{fernholz2010optimal} and \cite{ichiba2020relative} that having a non-trivial relative arbitrage amounts to showing that the auxiliary process
\[
d\zeta_i(t) = \zeta_i(t) \, dt + \sqrt{\zeta_i(t) \left( \zeta_1(t)  + \dots + \zeta_n(t) \right)} \, dW_i(t), \quad t \ge 0 
\]
may hit the boundary of the domain $[0, \infty)^n$ in finite time. We implement the $n$ dimensional auxiliary process $(\zeta_1(\cdot), \ldots, \zeta_n(\cdot))$ and visualize the trajectories in Figure~\ref{figuxbessel}(c), which shows that the auxiliary processes touch the boundary of the domain $[0, \infty)^n$ with positive probability.
% \begin{figure}
%     \centering
%     \includegraphics[width=0.35\linewidth]{zeta.png} \quad
%     \includegraphics[width=0.35\linewidth]{zeta4.png}
%     \caption{Sample trajectories of the $n$ dimensional auxiliary process $(\zeta_1(\cdot), \ldots, \zeta_n(\cdot))$ over time horizon $[0,1]$. The sample trajectories that touch the boundary $\{x_i = 0\}_{i=1}^n$ is marked red. The trajectories that stay positive is marked blue. Left: The trajectories for $n=2$. Right: The trajectories for $n=4$.}
%     \label{fig:aux}
% \end{figure}

\begin{remark}
    For a more general class of volatility-stabilized markets, let us consider $0 \leq \zeta \leq 1$.
     % $g(\mathbf{x}) = \frac{x_1 + \ldots x_n}{(x_1 \cdot x_n)^{-\kappa}}$, $k(\mathbf{x}) = (1- \zeta^2)(x_1 + \ldots + x_n) \sum_{j=1}^n \frac{1}{8 x_j} \geq 0$, 
     % \[
     % G(T, \mathbf{x}) = \mathbb{E}^{\mathbf{x}}\left[ \frac{X_1(T) + \ldots + X_n(T)}{(X_1(T) \cdot X_n(T))^{-\kappa}} e^{-\int_t^T (1- \zeta^2)(X_1(s) + \ldots + X_n(s)) \sum_{j=1}^n \frac{1}{8 X_j(s)} ds} \right].
     % \]
For $\mathcal{X}(t) = \mathbf{x} := (x_1, \ldots, x_n)$, the optimal arbitrage can be derived from the PDE characterization in \cite{fernholz2010optimal} that
     \[
     \begin{aligned}
     & u(T-t, \mathbf{x}) \\
     =  & \frac{(x_1 \cdots x_n)^{\frac{1+\zeta}{2}}}{x_1 + \ldots x_n} \mathbb{E}^{\mathbf{x}}\left[ \frac{X_1(T) + \ldots + X_n(T)}{(X_1(T) \cdots X_n(T))^{\frac{1+\zeta}{2}}} e^{-\int_t^T (1- \zeta^2)(X_1(s) + \ldots + X_n(s)) \sum_{j=1}^n \frac{1}{8 X_j(s)} ds} \right].
     \end{aligned}
     \]
This can be implemented similarly using the numerical approach in this section. In the next section, we propose a more general approach to deal with the PDE characterization \eqref{cauchypde}-\eqref{cauchypdeu0} of the optimal arbitrage.
\end{remark}

\section{A general approach via reflected Backward Stochastic Differential Equations}
\label{sec: rbsde}

The nonnegative minimal solution of PDEs is not only of interest in optimal arbitrage problem, as we have introduced in Section~\ref{sec: intro}. In this section, we discuss a different approach based on reflected Backward Stochastic Differential Equations (rBSDEs) to solve the minimal solution of \eqref{cauchypde}-\eqref{cauchypdeu0} that is not limited to volatility-stabilized market. 
% , by forcing the solutions to stay above a given stochastic process, called the obstacle. 
% This method can be applied to a more general class of differential equations whose nonnegative minimal solution is of interest, with numerous applications in finance and engineering, as detailed in .

Consider the PDE below
  \[
  \min [ u(\tau,\mathbf{x}) - h(\tau,\mathbf{x}), \partial_{\tau} u - \mathcal{A} u(\tau,\mathbf{x})] = 0, \quad (\tau,\mathbf{x}) \in (0,T) \times \mathbb{R}^n.
  \]
  \[
  u(0, \mathbf{x}) = g(\mathbf{x}), \quad \mathbf{x} \in \mathbb{R}^n.
  \]
We can write this equivalently as
  \begin{equation}
  \label{constraintreflect}
        \begin{aligned}
    0 &= \partial_\tau u - \mathcal{A}u(\tau, \mathbf{x}),   & \quad \{(\tau, \mathbf{x}) : u(\tau, \mathbf{x}) > h(t,\mathbf{x})\}, \\
    u(\tau, \mathbf{x}) & \geq h(\tau, \mathbf{x}), & \quad (\tau, \mathbf{x}) : (0,T) \times \mathbb{R}^n, \\
    u(\tau, \mathbf{x}) & \in C^{1,2}([0,T] \times \mathbb{R}^n), &\quad \{(\tau, \mathbf{x}) : u(\tau, \mathbf{x}) = h(\tau, \mathbf{x})\},\\
    u(0, \mathbf{x}) & = g(\mathbf{x}), & \quad \mathbf{x} \in \mathbb{R}^n.
  \end{aligned}
  \end{equation}
We denote by $u^r$ a solution to \eqref{constraintreflect}. In contrast, the non-negative solution of a parabolic PDE of our interest is 
\begin{equation}
\label{constraintcauchy}
      \begin{aligned}
    0 &= \partial_\tau u - \mathcal{A}u(\tau, \mathbf{x}),   & \quad (\tau, \mathbf{x}) : (0,T) \times \mathbb{R}^n, \\
    u(\tau, \mathbf{x}) & \geq h(\tau, \mathbf{x}), & \quad (\tau, \mathbf{x}) : (0,T) \times \mathbb{R}^n, \\
    u(\tau, \mathbf{x}) & \in C^{1,2}([0,T] \times \mathbb{R}^n), &\quad (\tau, \mathbf{x}) : (0,T) \times \mathbb{R}^n,\\
    u(0, \mathbf{x}) & = g(\mathbf{x}), & \quad \mathbf{x} \in \mathbb{R}^n.
  \end{aligned}
\end{equation}

% Since we are looking for the minimal solutions satisfying the above \eqref{constraintreflect} and \eqref{constraintcauchy} correspondingly, we have optimization problems
% \begin{equation}
% \label{minproblem1}
% \begin{aligned}
%     & \min u(\tau, \mathbf{x}) \quad \text{subject to \eqref{constraintreflect}}.
% \end{aligned}
% \end{equation}
% and 
% \begin{equation}
% \label{minproblem2}
% \begin{aligned}
%     & \min u(\tau, \mathbf{x}) \quad \text{subject to \eqref{constraintcauchy}}.
% \end{aligned}
% \end{equation}

Let $\mathcal{D}_1$ be the domain of the optimization problem $\min u^r(\tau, \mathbf{x})$, subject to \eqref{constraintreflect}, $\mathcal{D}_2$ be the domain of optimization problem $\min u(\tau, \mathbf{x})$, subject to \eqref{constraintcauchy}. We see that $\mathcal{D}_1 \subset \mathcal{D}_2$. Denote $u^{r \star}(\tau, \mathbf{x}) := \min u^r(\tau, \mathbf{x})$ and $u^{\star}(\tau, \mathbf{x}) := \min u(\tau, \mathbf{x})$. If both problems are feasible, then $u^{r \star}(\tau, \mathbf{x}) \le u^{\star}(\tau, \mathbf{x})$.

\begin{prop}
\label{prop832}
Suppose $u \in C^{1,2}([0,T] \times \mathbb{R}^n)$ is a solution of \eqref{cauchypde}. Define $\mathcal{X}(t) := (X_1(t), \ldots, X_n(t))$, for each $X_i(t)$, $i = 1, \ldots, n$,
\[
dX_i(t) = b_i(t) dt + \sum_{k=1}^n s_{ik}(t) dW_k(t), \quad t \ge 0 . 
\] 
% \[
% dX_i(t) = X_i(t) \beta_i(t) dt + X_i(t) \sum_{k=1}^n \sigma_{ik}(t) dW_k(t).
% \] 
Then $\{\mathcal{F}_t\}_{t \in [0,T]}$-adapted processes $\{\mathcal{X}_r, Y_r, Z_r\} := \{\mathcal{X}_r, u(T-r,\mathcal{X}_r), (s\nabla u ) (T-r,\mathcal{X}_r)\}$ solves 
\begin{equation}
\label{bsdeopt}
    u(T-t, \mathcal{X}_t) = u(0)-\int_t^T f(\mathcal{X}_r, Y_r, Z_r) dr - \int_t^T (Z_r)' dW(r),
\end{equation}
where \[
    \begin{aligned}
f(\mathcal{X}_t, Y_t, Z_t) 
% & =  b(\mathbf{x}) (s^T(t, \mathcal{X}_t))^{-1} s^{T}(t, \mathcal{X}_t) (D_x u(\tau, \mathbf{x}))^T - \frac{1}{x_1 + \ldots +x_n} s \cdot s^T(t, \mathcal{X}_t) (D_x u(\tau, \mathbf{x}))^T  \\
 & =  b(\mathbf{x}) (s'(t, \mathcal{X}_t))^{-1} Z_t - \frac{1}{x_1 + \ldots +x_n} \mathbf{1}' s(t, \mathcal{X}_t) Z_t.
\end{aligned}
\]
\end{prop}

\begin{proof}
By It\^o's formula on $u(\tau, \mathbf{x})$, $\tau:= T-t$, it follows
\[
du(T-t, \mathcal{X}(t)) = (\mathcal{L}u - \frac{\partial u}{\partial \tau})(T-t, \mathcal{X}(t)) dt + \sum_{k=1}^n R_k(T-t, \mathcal{X}(t)) dW_k(t),
\]
where 
% $R(\tau, \mathbf{x})$ is $n$-dimensional vector with 
$R_k(\tau, \mathbf{x}) = \sum_{i=1}^n x_i s_{ik}(\mathbf{x}) D_{i} u(\tau, \mathbf{x})$, $k = 1, \ldots, n$. $\mathcal{L}$ is the infinitesimal generator for $\mathcal X(\cdot)$, i.e.,  % of $(\mathbf{x}) \in (0,\infty)^n \times (0,\infty)^n$, i.e.,
% \begin{equation}
% \label{loperator}
\[
    \begin{aligned}
\mathcal{L} u(\tau, \mathbf x) = & \, \sum_{i=1}^n \sum_{j=1}^n \Big( b(\mathbf{x}) \cdot D_i u(\tau, \mathbf x) + \frac{1}{2} \text{tr} \big[a (\mathbf{x}) \cdot D^2_{ij} u(\tau, \mathbf x) \big] \Big) \Big \vert_{\mathbf x = \mathcal X(t) }.
\end{aligned}
\]
Substitute $\frac{\partial u}{\partial \tau}(\tau, \mathbf{x})$ with the Cauchy problem \eqref{cauchypde}, we get that the minimal non-negative continuous solution of the above equation satisfies $u(0) = 1$,
\[
du(T-t, \mathcal{X}(t)) = f(\mathcal{X}(t), u, Du) dt + \sum_{k=1}^n R_k(T-t, \mathcal{X}(t)) dW_k(t),
\]
where $f(\mathbf{x}, u, Du) = \sum_{i=1}^n b_i(\mathbf{x}) D_{i} u(\tau, \mathbf{x}) - \sum_{i=1}^n \sum_{j=1}^n \frac{a_{ij}(\mathbf{x}) D_{i} u(\tau, \mathbf{x})}{x_1 + \ldots +x_n}$. 
% $Z_r := R(T-r, \mathbf{x})$. 
For $\forall t \leq T$, integrate $du(T-r, \mathcal{X}(r))$ with respect to $r$, where $r \in [t,T]$,
we get that $\{\mathcal{X}_r, u(r,\mathcal{X}_r), (s \nabla u)(r,\mathcal{X}_r) \}$ solves \eqref{bsdeopt}.
\end{proof}
% where $u \in C^{1,2}([0,T] \times \mathbb{R}^n)$

Now, we can approach the minimal nonnegative solution of the Cauchy PDEs in \cite{fernholz2010optimal}, \cite{ichiba2020relative}, \cite{yang2023relative} by considering the following obstacle problem,
\begin{equation}
\label{refbsde}
    Y_{t} = g(\mathcal{X}_T) + \int_t^T f(\mathcal{X}_s,Y_s,Z_s)ds - \int_t^T Z_s dW_s + K_T - K_t, t \leq s \leq T,
\end{equation}
where $S_t = h(\tau, \mathcal{X}_t)$, $\tau = T-t$ is a continuous obstacle process that satisfies $\mathbb{E} [\sup_{0 \leq t \leq T} S_t^2] < \infty$ and $S_T$ is bounded almost surely.
$\{K_t\}$ is continuous nondecreasing predictable process, such that $K_0 = 0$, $\int_0^T (Y_t - S_t) dK_t = 0.$
% \begin{equation}
% \label{ytst}
%     \int_0^T (Y_t - S_t) dK_t = 0.
% \end{equation}
This acts as a minimal push, as the push occurs only
when the constraint is reached $Y_t = S_t$. Given the obstacle process $S_t$, the viscosity solution of the obstacle problem \eqref{constraintreflect} is shown in \cite{el1997reflected} to be equivalent to the reflected BSDE \eqref{refbsde}. 

We can take $S_t = 0$ and use the penalization method in \cite{el1997reflected} to approximate the minimal solution of \eqref{refbsde}, $(Y, Z, K)$, in the sense that for any other solution $(\Tilde{Y}, \Tilde{Z}, \Tilde{K})$, $Y \leq \Tilde{Y}$. 
% That is, for $t \in [0,T]$, the minimal solution of \eqref{refbsde} with $Y_t \geq S_t$ coincides with the solution of \eqref{refbsde} with $Y_t \geq S_t$. 
% We can investigate the numerical solutions of the obstacle problem \eqref{constraintreflect} and the Cauchy problem with the non-negativity constraint of solutions \eqref{constraintcauchy}. 
% Then, we can consider the reflected BSDE, that is, the constraint of the lower bound $u(\tau) = Y_t \geq S_t$, $0 \leq t \leq T$.

\section{Discussion}

Numerical approaches to high-dimensional PDEs typically requires the PDEs to have a unique solution. However, the existence of multiple solutions of the PDEs in Stochastic Portfolio Theory literature is essential, and solving the nonnegative minimal solution receives particular interests. In this paper, we investigate a probabilistic approach to solve the nonnegative minimal solution of a Cauchy PDE in the volatility-stabilized market. By formulating the stock capitalizations as time-changed squared Bessel processes, we solve the optimal arbitrage quantity numerically and demonstrate its regularity in time and space empirically. The extension of this approach to a more general market setup (\cite{ichiba2020relative, yang2023relative}) as an interacting particle system is related to numerical methods for stochastic differential games. Here, a possible next step is to solve the associated obstacle problem in the vein of Section~\ref{sec: rbsde} based on the corresponding forward-backward PDEs (\cite{achdou2020mean}). 
% The finite-difference scheme is first introduced by \cite{achdou2010mean}, and is extended in \cite{achdou2020mean} with both the states and controls in the mean field interaction term. In general, a promising next step is to look for an efficient deep learning approach to tackle differential equations with multiple solutions.

\bibliographystyle{abbrv}
\bibliography{main}

\end{document}